\documentclass[aps,pra,twocolumn,amsmath,amssymb,nofootinbib,showpacs,superscriptaddress]{revtex4-1}
\usepackage[english]{babel}
\usepackage{latexsym}
\usepackage{graphics}
\usepackage{graphicx}
\usepackage{epsfig}
\usepackage{color}
\usepackage{bm}
\usepackage{amsmath}
\usepackage{amssymb}
\usepackage{amsthm}
\usepackage{dcolumn}
\usepackage{bm}
\usepackage{float}
\usepackage{hyperref}
\usepackage{color}
\usepackage{epstopdf}
\usepackage{cleveref}
\usepackage[svgnames]{xcolor}
\usepackage{enumerate}
\usepackage{braket}
\hypersetup{hidelinks,colorlinks=true,allcolors=DarkBlue}

\theoremstyle{remark}
\newtheorem{remark}{Remark}
\theoremstyle{plain}
\newtheorem{theorem}{Theorem}
\newtheorem{corollary}{Corollary}

\begin{document}

\preprint{APS/123-QED}

\title{Quantum key distribution protocol with pseudorandom bases}

\author{A.S. Trushechkin}
\affiliation{Steklov Mathematical Institute of Russian Academy of Sciences, Moscow 119991, Russia}
\affiliation{National Research Nuclear University ``MEPhI'', Moscow 115409, Russia}
\affiliation{Department of Mathematics and Russian Quantum Center, National University of Science and Technology MISiS, Moscow 119049, Russia}

\author{P.A. Tregubov}
\affiliation{National Research Nuclear University ``MEPhI'', Moscow 115409, Russia}

\author{E.O. Kiktenko} 
\affiliation{Steklov Mathematical Institute of Russian Academy of Sciences, Moscow 119991, Russia}
\affiliation{Russian Quantum Center, Skolkovo, Moscow 143025, Russia}

\author{Y.V. Kurochkin}
\affiliation{Russian Quantum Center, Skolkovo, Moscow 143025, Russia}
\affiliation{Department of Mathematics and Russian Quantum Center, National University of Science and Technology MISiS, Moscow 119049, Russia}

\author{A.K. Fedorov}
\affiliation{Russian Quantum Center, Skolkovo, Moscow 143025, Russia}
\affiliation{Department of Mathematics and Russian Quantum Center, National University of Science and Technology MISiS, Moscow 119049, Russia}
\affiliation{LPTMS, CNRS, Univ. Paris-Sud, Universit\'e Paris-Saclay, Orsay 91405, France}

\date{\today}
\begin{abstract}
Quantum key distribution (QKD) offers a way for establishing information-theoretically secure communications. 
An important part of QKD technology is a high-quality random number generator (RNG) for quantum states preparation and for post-processing procedures.
In the present work, we consider a novel class of prepare-and-measure QKD protocols, 
utilizing additional pseudorandomness in the preparation of quantum states.
We study one of such protocols and analyze its security against the intercept-resend attack.
We demonstrate that, for single-photon sources, the considered protocol gives better secret key rates than the BB84 and the asymmetric BB84 protocol. 
However, the protocol strongly requires single-photon sources.
\end{abstract}
\maketitle

\section{Introduction}

Quantum algorithms exploit the laws of quantum mechanics to solve problems exponentially faster than their best classical counterparts~\cite{Chuang}. 
Shor's quantum algorithm for fast number factoring attracted a great attention since this problem is in heart of public-key cryptosystems~\cite{Shor}.
In  view of the Shor's algorithm, the only way to ensure the absolute long-term security is to use information-theoretically secure primitives, 
such as the one-time pad scheme~\cite{Vernam,Shannon,Schneier}.
However, the need for establishing secret keys between communicating parties invites the challenge of how to securely distribute these keys~\cite{Schneier}.

Fortunately, together with the tool for breaking public-key cryptographic primitives, quantum physics allows one to establish secure communications~\cite{BB84,BB842,BB843,Gisin,Scarani}. 
By encoding information in quantum states of photons, transmitting them through fiber channels and communication via authenticated classical channel, 
QKD systems offer a practical tool for private key distribution~\cite{BB84,BB842,BB843,Gisin,Scarani}.
Unlike classical cryptography, QKD promises information-theoretical security based on the quantum physics laws~\cite{Gisin,Scarani,Wootters}. 
During last decades, great progress in theory, experimental study, and technology of QKD has been performed~\cite{LoRevs}.
However, QKD technology faces a number of challenges such as distance, key generation rate, practical security, and many others~\cite{LoRevs}. 

The idea behind the seminal proposal for QKD protocol, known as BB84 protocol~\cite{BB84}, is inspired by conjugate coding~\cite{Wiesner}.
The BB84 protocol employs the idea of usage of two orthogonal polarizations states of photons.
The BB84 protocol has been widely studied~\cite{BB84,BB842,BB843,Gisin,Scarani}, 
and its security has been proven~\cite{Lo,Lutkenhaus,Shor2,Mayers,Renner,Gottesman}.
Development of novel QKD protocols, 
offering ways to push the performance of QKD technology, 
is on a forefront of quantum information technologies.
During last decades several extensions of the BB84 protocol and alternative QKD protocols, 
such as E91~\cite{E91}, B92~\cite{B92}, six-state BB84 protocol~\cite{Gisin5}, asymmetric BB84 (we will abbreviate it as aBB84)~\cite{aBB84}, SARG04~\cite{SARG}, 
differential-phase shift~\cite{DPS}, coherent one way~\cite{COW}, and also setups with continuous variables~\cite{CV}, have been actively discussed. 

\begin{figure}[t]
\begin{centering}
\includegraphics[width=1\columnwidth]{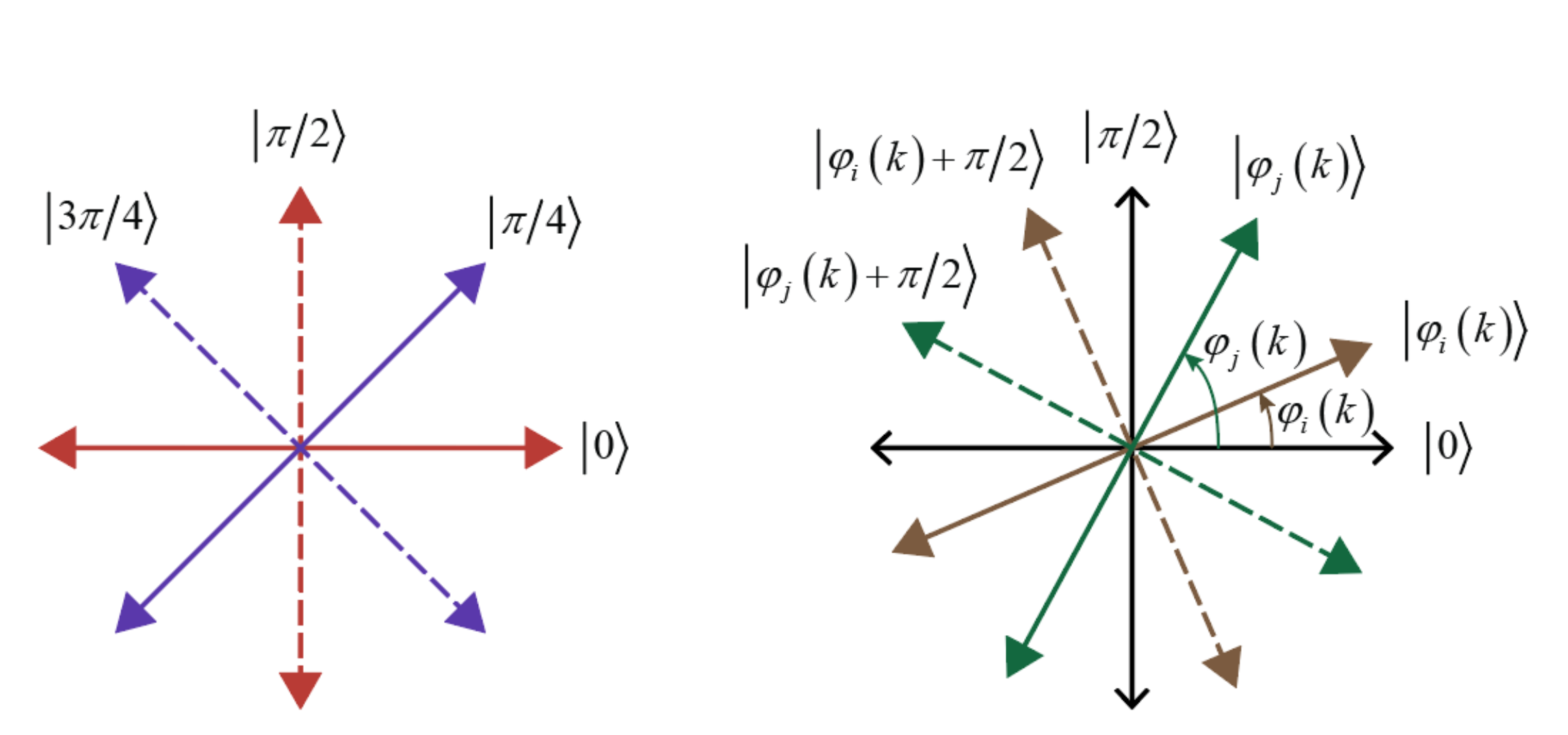}
\end{centering}
\vskip -4mm
\caption
{
Bases patterns on the Poincar\'e sphere.
The BB84 protocol (left) uses two maximally conjugated bases with the angle $\pi/4$ between each other. 
For each pulse, the bases are chosen by Alice and Bob randomly and independently, so, the sifting procedure (discarding of positions where Alice's and Bob's bases are different) is required. 
In the suggested protocol (right), the standard basis $\{\ket0,\ket1\}$ is rotated by an arbitrary angle (from a finite set) in not a random pseudorandom manner. 
Thus, the bases of Alice and Bob always coincide and there is no sifting.
}
\label{fig:setup}
\end{figure}

The point we want to stress here is the fact that for the seminal BB84 protocol and most of its variations, 
something should provide the ignorance of an eavesdropper (Eve) about the bases in which quantum states are encoded~\cite{BB84}. 
The BB84 protocol provides this condition by the random independent choice of the bases by legitimate parties (Alice and Bob). 
To this end, Alice and Bob use true random number generators (TRNG). 
However, the cost is the sifting procedure: 
Alice and Bob must discard the positions with incompatible basis choices. 
This leads to a loss of approximately a half of the raw key. 
In order to reduce the losses in the sifting procedure, the aBB84 protocol has been proposed~\cite{aBB84}. 
In this protocol, Alice and Bob use one basis with a high probability and a conjugate basis with a small probability. 
The first basis is used mainly to establish a secret key, while the second one is used to verify the absence of eavesdropping. 
We will refer to these bases as ``the signal basis'' and ``the test basis'', respectively.
In the asymptotic case of an infinitely large number of transmitted quantum states, the  probability of the use of the test basis can be made arbitrarily small. 
Hence,  the basis choices of Alice and Bob are almost always coincide and there is almost no sifting. 
Nevertheless, for a finite number of transmitted states, 
this probability cannot be made arbitrarily small since a reliable statistics for the test basis should be collected for tight estimation of the amount of eavesdropping~\cite{aBB84,Gisin6}.

In the present work,
we consider a novel class of QKD protocols, which utilizes the pseudorandomness in the preparation of quantum states. 
Namely, Alice and Bob can use not random, but pseudorandom sequence of bases generated from a common short secret key. 
On the one hand, their bases always coincide, so, the suggested scheme allows one to avoid the sifting procedure. 
On the other hand, for Eve, who does not know this key, the sequence is similar to a random one and she cannot predict it.
On the basis of this idea, 
we study a protocol with pseudorandom choice of bases (abbreviated as PRB) and analyze its security against the intercept-resend attack.
The suggested protocol is a formalization of a protocol (the floating basis protocol) described by one of the authors~\cite{Kurochkin2005,Kurochkin2009,Kurochkin2015}.
In this work, we assume that the sequence of logical bits is truly random, but the sequence of bases is pseudorandom. 
We then demonstrate that the PRB protocol gives higher secret key rates than the BB84 protocol and approximately the same rates as the asymmetric BB84 protocol. 
However, the PRB strongly requires single-photon sources.

A general motivation for the development of new QKD protocols is exploration of different ways of how we can exploit the properties of quantum information to provide information security. 
The novel idea of the protocol proposed here is a method of combining of classical pseudorandomness with quantum encoding of information. 
We should note that the known Y00 protocol~\cite{Y00,Y002,Y003} also use pseudorandom quantum states. It
provides a randomized stream cipher with information theoretic-security by a randomization based on quantum noise and additional tools.
Another important example of utilizing pseudorandomness is recently suggested mechanism for quantum data locking~\cite{Lloyd016,Pan2016}. 

We can treat QKD protocols utilizing pseudorandomness in such a way. 
Since generators of true random numbers are not sufficiently fast, pseudorandom number generators (PRNGs) are used in practical setups instead~\cite{Zbinden2014}. 
It is interesting to study how the use of pseudorandom numbers instead of true random numbers affects the security of QKD protocols (see Ref.~\cite{Bouda2012}) and, 
moreover, can it be even advantageous. 
Here we assume that the sequence of logical bits is truly random, but the sequence of bases is pseudorandom. 
As it was explained above, if we are able to prove the security of such scheme, we can made it more advantageous due to avoiding the sifting 
(with the cost of an additional secret key consumption for initial secret random seed for the PRNG in future sessions).

The paper is organized as follows.
The new QKD protocol, which we will refer to as pseudorandom basis (PRB) protocol, is described in Sec.~\ref{sec:protocol}.
Its security against the intercept-resend attack is proved in Sec.~\ref{sec:attack}.
We summarize the main results of our work in Sec.~\ref{sec:conclusion}. 
In Sec.~\ref{sec:pns}, we analyze the photon number splitting (PNS) attack and show that, unlike the BB84 protocol and its modifications, 
the PRB protocol strongly requires a single-photon source of light.

\section{QKD protocol with pseudorandom bases}\label{sec:protocol}

Let Alice and Bob have a common pre-shared key 
\begin{equation}
	k=(k_1,\ldots,k_l)\in\{0,1\}^l,
\end{equation}
which is the seed for the pseudorandom number generator (PRNG), $l$ is the key size. 
We use the following notation: 
\begin{equation}
	\ket\varphi=\cos\varphi\ket0+\sin\varphi\ket1,
\end{equation}
where $\{\ket0,\ket1\}$ is the standard basis. 
As usual in QKD, Alice and Bob have quantum channel and authenticated public classical channel: Eve freely read the communication over this channel, but cannot interfere in it.

The considered class of QKD protocol based on pseudorandomness operates as follows. 

\begin{enumerate}[(i)]

\item Using the common pre-shared key $k$ and the PRNG, Alice and Bob generate a common pseudorandom sequence in the following form:
\begin{equation}\label{eq:sequence} 
	\varphi_1(k),\ldots,\varphi_N(k), 
	\quad  
	\varphi_i(k)\in\left\{\frac{\pi j}{2M}\right\}_{j=0}^{M-1}
\end{equation}
where $M=2^m$ for some $m\geq1$. Schematically such pseudorandom rotations of the standard basis $\{\ket0,\ket1\}$ are shown on Fig.~\ref{fig:setup}.
We assume that $l$ is divisible by $m$ and denote $l/m=l'$. 
We assume that $N=2^{l'}$.

\item Using a TRNG, Alice generates the random bits as follows:
\begin{equation}\label{eq:alice-bits} 
	x_1,\ldots,x_N.
\end{equation}

\item Using sequence~(\ref{eq:sequence}) and generated random bits~(\ref{eq:alice-bits}), Alice prepares the following sequence of states:
\begin{equation}
	\ket{\varphi_1(k)+x_1\pi/2},\ldots,\ket{\varphi_N(k)+x_N\pi/2},
\end{equation}
and sends them to Bob over the quantum channel. 

\item Bob measures them using the following bases:  
\begin{equation}
	\{\ket{\varphi_i(k)},\ket{\varphi_i(k)+\pi/2}\}, \quad i=1,\ldots,n.
\end{equation}

\item Bob then writes the results of these measurements in the binary variables $y_i$ as follows: 
$\ket{\varphi_i(k)}$ corresponds to $y_i=0$ and $\ket{\varphi_i(k)+\pi/2}$ corresponds to $y_i=1$.
In the case of ideal channel and no eavesdropping, $x_i=y_i$ for all $i$, hence, Alice and Bob can use their binary strings $x$ and $y$ as a common secret key. 
Due to noise in the channel and, probably, eavesdropping, there are some errors in these strings, and Eve potentially have some information about them. 
They are, thus, called the {\it raw keys}. 

\item The following steps of post-processing of raw keys coincide with those of BB84, so, we only briefly mention them: 
Alice and Bob perform the error correction (using error-correcting codes or interactive error correction protocols; 
for the last issues concerning the adaptation of error-correcting codes for QKD see Ref.~\cite{Kiktenko2017}) and calculate the number of detected errors~\cite{Note1,Fedorov2016}. 
If the number of errors exceeds a certain threshold, which make the secret key distribution impossible, Alice and Bob abort the protocol. 
Otherwise, they perform privacy amplification to reduce the potential Eve's information to a negligible level. 
The resulting key is called the \textit{secret key} or the \textit{final key}. 
It is the output of our QKD protocol.

\end{enumerate}

The underlying idea of the protocol is as follows. 
If Eve does not know the initial secret key $k$, and the pseudorandom angles $\varphi_i$ are similar to truly random, 
then she cannot guess all the bases correctly and her eavesdropping will cause disturbance in the Alice's and Bob's raw keys. 
Of course, a rigorous analysis is required since which takes into account that the sequence $\{\varphi_i(k)\}_{i=1}^N$  it not truly random but pseudorandom.

\begin{remark}
We note that every quantum key distillation protocol that makes use of pre-shared key can be transformed into 
an equally efficient protocol which needs no pre-shared key~\cite{Renner2007}. 
However, this fact is irrelevant for our case since we use the pre-shared key not on the stage of secret key distillation from raw keys, 
but on the previous stage of transfer of quantum states. 
\end{remark}

\begin{remark}
If $M=2$, then the protocol uses two BB84 bases. In this case, the protocol can be called  ``BB84 with pseudorandom sequence of bases'', 
while the case $M>2$ can be called the ``multibasis protocol''. 
We will see that the multibasis version gives higher secret key rates due to additional uncertainty for Eve.

The possibility of the use of arbitrary number of bases is a consequence of their correlated choice by Alice and Bob; 
otherwise this would lead to an increase of the number of positions with inconsistent bases. 

\end{remark}

\subsection{PRNG based on the Legendre symbol}\label{SecLegendre}

Our choice for the PRNG is the Legendre symbol PRNG, since it provides an almost uniform distribution of patterns, which will be exploited in the security proof. 
The PRNG is defined as follows.

Let $L$ be a prime number (public value) such that $L\equiv 3\pmod 4$, and $k\in[0,L-1]$ be a secret key. 
Let us then define
\begin{equation}\label{EqLegendre}
	\overline a_i=
	\begin{cases}
		1,&\text{if $i$ is a quadratic residue modulo $L$}\\
		 &\text{and}\, i\not\equiv0\pmod L; \\
		0,&\text{otherwise},
	\end{cases}
\end{equation} 
\begin{equation}\label{EqLegendrePRNG}
	a_i(k)=\overline a_{k+i}.
\end{equation}
Recall that $x$ is called a quadratic residue modulo $L$ if there exists an integer $y$ such that $y^2\equiv x\pmod L$. 
If $i\not\equiv0\pmod L$, the value $2\overline a_i-1$ is called the Legendre symbol of $i$. We will refer to the sequence 
\begin{equation}
	\overline a_1,\overline a_2,\ldots
\end{equation}
as the Legendre sequence. It is periodic with the period $L$.

Pseudorandom properties of Legendre sequences are known for a long time~\cite{Aladov1896}. 
In particular, the distribution of patterns of Legendre sequences is known to be close to uniform~\cite{Aladov1896,Moroz,Ding1998} (for details, see property (\ref{EqLegPatt}) from Appendix~A
and, as its direct consequence, property (\ref{EqPatt}) from Appendix~C). 
This will be important for estimation of the number of bases correctly guessed by Eve.

Let us specify the use of this PRNG for our protocol. In the two-basis version of the protocol (BB84 with pseudorandom sequence of bases), 
the basis for each position $i$ is specified by a single bit $a_i(k)$. 
The length of the key $l$ is the number of bits required to specify $k$, i.e., $\lceil \log_2 L\rceil$. 
In the following, $\lceil x\rceil$ and $\lfloor x\rfloor$ denote the ceil and the floor of $x$  (the closest integer to $x$ from above and from below), respectively.

In the multibasis version of the protocol, every basis is specified by $m$ bits, or $m$ registers. 
Each register has its own PRNG based on the Legendre symbol, so that
\begin{equation}\label{EqRndRegs}
	\varphi_i(k)=\frac\pi2\sum_{j=1}^ma_i(k^{(j)})2^{-j},
\end{equation}
where $k^{(j)}$ is a subkey (of length $l'=\lceil \log_2 L\rceil$) for the $j$th register and the sequence $(a_1(k^{(j)}),a_2(k^{(j)}),\ldots)$ is specified by (\ref{EqLegendrePRNG}). 
The total key $(k^{(1)},\ldots,k^{(m)})$ has the length $l=l'm$

\section{Intercept-resend attack}\label{sec:attack}

The simplest attack on BB84-like protocols is the intercept-resend attack. 
Here we describe this attack for the considered class of QKD protocols. 

\begin{enumerate}[(i)]

\item
Eve chooses some positions $1\leq i_1<\ldots<i_n\leq N$ to intercept, where $0<n\leq N$. 
Denote $\gamma=n/N$ the fraction of positions she intercepts. Then, for each $j=1,\ldots,n$, Eve performs the next steps.

\item 
 Eve chooses an angle $\beta_{i_j}$, measures the $i_j$th qubit in the basis
\begin{equation}
	\left\{\ket{\beta_{i_j}},\ket{\beta_{i_j}+\frac\pi2}\right\}, 
\end{equation}
and writes the result in the variable $z_{i_j}$ (0 or 1, respectively).

\item
 Eve sends a new qubit in the state $\ket{\beta_{i_j}+ z_{i_j}\pi/2}$ to Bob. 

\end{enumerate}

The crucial point is that the results of Eve's measurements alone leaks no information about the bases and, hence, about the initial secret key (the seed for the PRNG) $k$.

Indeed, denote $p(k|\mathbf z)$ the probability of the key $k$ conditioned on the Eve's results of the measurements $\mathbf z=(z_{i_1},\ldots,z_{i_n})$. 
Denote
\begin{equation}
	\bm\varphi=\bm\varphi(k)=(\varphi_1(k),\varphi_2(k),\ldots,\varphi_N(k)).
\end{equation}
We then define the probability 
\begin{equation}
	p(k|\mathbf z)=p(\bm\varphi|\mathbf z)=\frac{p(\mathbf z|\bm\varphi)p(\bm\varphi)}{p(\mathbf z)}.
\end{equation}
Here
\begin{equation}
	p(\mathbf z|\bm\varphi)=\prod_{j=1}^n p(z_{i_j}|\varphi_{i_j})=2^{-n}
\end{equation}
since
\begin{eqnarray}
	&&p(z|\varphi)=p(z|\varphi,x=0)p(x=0)+p(z|\varphi,x=1)p(x=1)\nonumber\\
	&&=\frac{1}{2}[p(z|\varphi,x=0)+p(z|\varphi,x=1)]\\
	&&=\frac{1}{2}\left[\cos^2\left(\beta-\varphi-\frac\pi2 z\right)+\sin^2\left(\beta-\varphi-\frac\pi2 z\right)\right]=\frac{1}{2}\nonumber,
\end{eqnarray}
where we have omitted the subindices $i_j$. 
We have also
\begin{equation}
	p(\mathbf z)=\sum_{\bm\varphi'}p(\mathbf z|\bm\varphi')p(\bm\varphi')=2^{-n}\sum_{\bm\varphi'}p(\bm\varphi')=2^{-n}.
\end{equation}
Therefore, we arrive at the following expression: 
\begin{equation}\label{EqNoAprioriInf}
	p(k|\mathbf z)=p(\bm\varphi|\mathbf z)=p(\bm\varphi)=p(k),
\end{equation}
which means that the a posteriori probability is equal to the a priori one. 
In other words, Eve obtains no information on the key if the knows only the results of her measurement. 
This is due to the randomization introduced by the random and yet unknown to Eve bits $x_i$. 
In fact, the quantum state of a qubit for unknown $x$ is the completely mixed state:
\begin{equation}\label{EqChaoticState}
	\frac{1}{2}\ket{\varphi}\bra{\varphi}+\frac{1}{2}\ket{\varphi+\frac\pi2}\bra{\varphi+\frac\pi2}=\frac{1}{2}I,
\end{equation}
where $I$ is the identity operator. 

Thus, on the stage of quantum state transmission, Eve chooses the angle $\beta_{i_j}$ with no information on the key and have to \textit{guess} the bases or the initial key. 
Since it is unlikely that she correctly guesses all bases, we arrive at the keystone of the security of QKD: eavesdropping cause disturbance. 
Rigorous estimations of the number of bases that Eve can correctly guess is the main part of security proof.

From the other side, we assume that, after the accomplishment of all stages of the protocol and, moreover, after the transmission of the encrypted message, 
Eve is able to determine the initial key. So, a posteriori, she gets knowledge of the correct bases. 
In Appendix~B we show that Eve needs of order $l$ bits to intercept to guess the initial key if she knows the encrypted message (``known plaintext attack''). 

In our analysis, we assume that $N$ and $n$ are so large that we can neglect the statistical fluctuations, since our aim is to give general analysis of the protocol, 
not the ultimate formulas for the practical applications. 

\subsection{BB84 with pseudorandom sequence of bases}\label{SecBB84PRB}

For a transparent analysis, we first consider the protocol BB84 with pseudorandom bases. 
In this case, the basis choice is specified by a single bit $a_i$. 
Let the upper bound on the number of bases correctly guessed by Eve for a given  $\gamma$ be $n_{\rm correct}(\gamma)$. 
Respectively, $n_{\rm incorrect}(\gamma)=n-n_{\rm correct}(\gamma)$ is the lower bound on the number of incorrect guesses. 
Recall that we assume that  Eve eventually gets knowledge of the initial key $k$, hence, for each position, she knows whether she has correctly guessed the basis in this position or not. 
Consequently, the QBER ($q$) and the Eve's mean information on a raw key bit are as follows:
\begin{eqnarray}
	q(\gamma)&=&\frac12\frac{n_{\rm incorrect}(\gamma)}N \label{EqQalpha} \\
	I_{\rm E}(\gamma)&=&\frac{n_{\rm correct}(\gamma)}N.
\end{eqnarray}
The legitimate parties have the measured (or estimated) value of QBER. 
If we replace the left-hand side of Eq.~(\ref{EqQalpha}) by this value, we can find the inverse function as follows:
\begin{equation}
	\gamma=\gamma(q).
\end{equation}
This is an estimation of the fraction of qubits intercepted by Eve for a given QBER. 
Then one has
\begin{equation}
	I_{\rm E}(q)=\frac{n_{\rm correct}(\gamma(q))}N.
\end{equation}
From the other side, Bob's mean information on a bit of the Alice's raw key is 
\begin{equation}
	I_{\rm B}(q)=1-h(q).
\end{equation}
Here 
\begin{equation}
	h(p)=-p\log_2p-(1-p)\log_2(1-p)
\end{equation}
is the binary entropy, $0\leq h(p)\leq 1$. 
However, to fully exploit this information, 
Alice and Bob require error-correcting scheme that achieves the theoretical (Shannon) limit, in which $h(q)$ bits of information about raw keys are revealed over the public channel. 
Practically, $f(q)h(q)$ bits are revealed, where $f(q)\geq1$ is the efficiency of the error correction scheme. So, the ``effective'' Bob's mean information on a bit of the Alice's raw key is 
\begin{equation}
	I_{\rm B}(q)=1-f(q)h(q).
\end{equation}

Then, the secret key rate (per transmitted qubit, also called secret fraction) has the following form~\cite{Csiszar,Gisin,Scarani}:
\begin{equation}\label{EqRate}
\begin{split}
	\!\!R(q)=I_{\rm B}(q)-I_{\rm E}(q)=1-f(q)h(q)-I_{\rm E}(q).
\end{split}
\end{equation}

Eve can try to guess the elements of the sequence $\{a_i\}$ as it were a truly random sequence. 
In this case, she correctly guesses approximately 
\begin{equation}
	n_{\rm correct}\approx \frac n2=\frac{\gamma N}2
\end{equation}
of the bases.
Thus, 
\begin{equation}
	q(\gamma)=\frac\gamma4, \,\, \gamma(q)=4q, \,\, I_{\rm E}(q)=2q, 
\end{equation}
and the secret fraction is as follows:
\begin{equation}\label{EqRateABB84}
\begin{split}
	R(q)=I_{\rm B}(q)-I_{\rm E}(q)=1-f(q)h(q)-2q.
\end{split}
\end{equation}

But Eve can exploit the fact that the sequence $\{a_i\}$ is not random, but pseudorandom and contains some regularities. 
The estimation of $n_{\rm correct}(\gamma)$ for this case is rather involved and is given in Appendix~C. 
Here we give a summary of the analysis and results of Appendix~C. 

The analysis of pseudorandom sequences is an important part of classical cryptography. 
But in classical cryptography it is usually assumed that Eve has a limited computing power and cannot use the brutal force attack. 
Here we assume that Eve has an unlimited computing power, which is common for quantum cryptography. 
Suppose that Eve succeeded to guess a subset $\mathcal K_1\subset \mathcal K$ which contains the actual key $k$. 
Other words, it is unlikely that she correctly guesses the key $k$ for large $l$ (the probability is $1/|\mathcal K|\sim 2^{-l}$), 
but she can guess that $\overline k$ belongs to a certain subset $\mathcal K_1$. 
The probability of such success is equal to $|\mathcal K_1|/|\mathcal K|$. 
Then she can choose not arbitrary $n=\gamma N$ positions to attack, but  special positions. 
Namely, positions $i$ such that the bits $a_i(k')$ coincide with each other for most $k'\in\mathcal K_1$ are preferable. 
Following this way of thinking, we arrive at the optimization problem. If $|\mathcal K_1|$ is less or comparable to $l$, we are able to solve it explicitly. 
This is done in Theorem~\ref{Th} and adopted for practical situation in Corollary~\ref{CorGuess}. In this case, we can use explicit formula (\ref{Eqncorropt}). 
If $|\mathcal K_1|$ is large, then we can still use formula (\ref{Eqncorropt}), but it gives too pessimistic (for Alice and Bob) estimate of $n_{\rm correct}(\gamma)$. 
A tighter bound can be obtain if we numerically solve the linear programming problem given in formula~(\ref{Eqncorroptlp}) (Corollary~\ref{CorGuess2}). 
The linear programming problems are known to have efficient algorithms of solutions. 

Of course, $n_{\rm correct}(\gamma)$ increase as $|\mathcal K_1|$ decreases (Eve adopt her attack to a tighter set of keys). 
However, the probability that $k\in\mathcal K_1$ is $|\mathcal K_1|/|\mathcal K|$, i.e., small whenever $|\mathcal K_1|$ is small. 
So, both estimates (\ref{Eqncorropt}) and (\ref{Eqncorroptlp}) are dependent on the additional parameter $\varepsilon=|\mathcal K_1|/|\mathcal K|$, 
i.e., $n_{\rm correct}=n_{\rm correct}(\gamma,\varepsilon)$. 
The parameter $\varepsilon$ can be called the failure probability: the probability that Eve will succeed to guess a more tight set containing the actual key, 
other words, that she will be more lucky than we expect. The emergence of such (in)security parameter is common for QKD security proofs \cite{Renner}.

In short, we use Eq.~(\ref{Eqncorropt}) (explicit formula) or Eq.~(\ref{Eqncorroptlp}) (linear programming problem which gives a tighter bound) to estimate $n_{\rm correct}(\gamma,\varepsilon)$ 
from above for given failure probability $\varepsilon$. 
These estimations can be substituted to Eq.~(\ref{EqQalpha}) to find the function $q(\gamma)$ and then to Eq.~(\ref{EqRate}) to obtain (numerically) the secret fraction.

It turns out that Eve can guess more elements of pseudorandom sequence than those of truly random sequence (see the end of Appendix~C). By this reason, the BB84 protocol with pseudorandom sequence of bases gives higher secret key rates than the usual BB84 protocol (because of the absence of sifting), but lower secret key rates than the asymmetric BB84 protocol. So, we do not consider the BB84 protocol with pseudorandom sequence of bases as a real alternative to aBB84 and switch to the multibasis case. In the multibasis case, the larger number of bases that Eve can correctly guess for the pseudorandom sequence is compensated by additional uncertainty for Eve caused by the use of many (instead of two) bases. In Sec.~\ref{SecNum}, we will compare the results of the multibasis PRB protocol with BB84 and aBB84 and show that the multibasis protocol can give slightly better results than the aBB84 protocol.

\subsection{Multibasis case}

Here we investigate the intercept-resend attack for the multibasis version if the protocol. 
Denote the difference between the Eve's guess of the $i$th angle $\varphi_i^{\rm E}$ and the actual angle $\varphi_i(k)$ as $\Delta_i$ and let $(b^{(1)}_i,\ldots,b^{(m)}_i)$ be its binary expansion:
\begin{equation}\label{EqBinExpanDelta}
	\Delta_i=\varphi_i^{\rm E}-\varphi_i(k)=\frac\pi2\sum_{j=1}^m  b_i^{(j)} 2^{-j}.
\end{equation}
For each register $j$, 
the upper bound of the number of bits $b^{(j)}_i$ correctly guessed by Eve is $n_{\rm correct}(\gamma)$ given by either Eq.~(\ref{Eqncorropt}) or Eq.~(\ref{Eqncorroptlp}) from Appendix~C 
(i.e., now $n_{\rm correct}(\gamma)$ denotes the number of correctly guessed bits in a single register). 
Denote $T\subset\{1,\ldots,N\}$ the set of pulses intercepted by Eve, $|T|=n=\gamma N$. 
Let us pick a position from $T$ at random. 
For each register, consider the event that the corresponding bit is correctly guessed. 
The probability of this event is (at most) $n_{\rm correct}(\gamma)/(\gamma N)$. Since the keys for different registers are chosen independently, these events are independent. 
Therefore, one has
\begin{eqnarray}
	&&\Pr\left[\Delta_i=\frac{\pi t}{2M}\right]\equiv p_t(\gamma)\nonumber\\
	&&=\prod_{j=0}^{m-1} \Pr\left[b_i^{(j)}=\lfloor 2^{-j}t\rfloor \bmod 2\right]
	\nonumber\\
	&&=\!\left(\frac{n_{\rm correct}(\gamma,\frac\varepsilon m)}{\gamma N}\right)^{\#0(t)}\!\left(\frac{n_{\rm incorrect}(\gamma,\frac\varepsilon m)}{\gamma N}\right)^{\#1(t)}\!,~
\end{eqnarray}
where ${\#0(t)}$ and ${\#1(t)}$ are the numbers of 0's and 1's in the binary expansion of $t$. 
Note the argument $\varepsilon/m$ of the function $n_{\rm incorrect}$: 
if the probability that Eve correctly guesses more then a given number of bits in a single register is not greater than $\varepsilon/m$, 
then the probability that Eve correctly guesses more then a given number of bits in each of $m$ registers is not greater than $\varepsilon$.

\begin{remark}\label{RemNcorr}
For Eve, the correct guessing of the highest-order bit $b_i^{(1)}$ in the binary expansion (\ref{EqBinExpanDelta}) is of the most importance. 
So, her optimal strategy is to chose positions to intercept which maximize the number of correctly guessed elements in the sequence for the first register $(b_1^{(1)},b_2^{(1)},\ldots)$. 
The maximal number of correct guesses is bounded from above as $n_{\rm correct}(\gamma,\varepsilon)$. 
Since Eve adjusts attack to optimize the number of correct guesses in the first register, she is not so good in the number of correct guesses in further registers. 
But, in favor of Eve, we bounded the number of correct guesses for other registers from above also by the same quantity $n_{\rm correct}(\gamma,\varepsilon)$.
\end{remark}

Now let us derive formulas for QBER and Eve's mean information on a raw key bit.
For simplicity, let us drop the subscript $i$. Denote $x\in\{0,1\}$ the bit value transmitted by Alice, $y,z\in\{0,1\}$ the results of Bob's and Eve's measurements. 
We then have 
\begin{widetext}
\begin{eqnarray}
	&&p(z|x,\Delta)=\cos^2\left(\Delta+\frac\pi2(x-z)\right),\quad p(y|z,\Delta)=\cos^2\left(\Delta+\frac\pi2(z-y)\right). \\
	&&p(y\neq x|\Delta)=p(y\neq x|z=x,\Delta)p(z=x|\Delta)+p(y\neq x|z\neq x,\Delta)p(z\neq x|\Delta)=\frac12\sin^2(2\Delta)=\frac14[1-\cos(4\Delta)], \\
	&&p(y\neq x)=\sum_{j=0}^{M-1}p_j(\gamma)p\left(y\neq x|\Delta=\frac{\pi j}{2M}\right)=\frac14\sum_{j=0}^{M-1}p_j(\gamma)\left[1-\cos\left(\frac{2\pi j}{M}\right)\right].
\end{eqnarray}
\end{widetext}
where $p(y\neq x)$ is the probability of error in Alice's and Bob's bit for an intercepted position. 

To obtain the QBER value, one has to multiply this quantity on the fraction of intercepted positions:
\begin{equation}\label{EqQmany}
	q=\gamma p(y\neq x)=\frac\gamma 4\sum_{j=0}^{M-1}p_j(\gamma)\left[1-\cos\left(\frac{2\pi j}{M}\right)\right].
\end{equation}
The Eve's information on an intercepted raw key bit $x$ is as follows:
\begin{equation}
	I_{\rm E}^{\rm intercepted}(\gamma)=1-\sum_{j=0}^{M-1}p_j(\gamma)h\left(\cos^2\left(\frac{\pi j}{2M}\right)\right)
\end{equation} 
The mean Eve's information on a raw key bit then has the following form:
\begin{equation}
	I_{\rm E}(\gamma)=\gamma I_{\rm E}^{\rm intercepted}(\gamma).
\end{equation}
From Eq.~(\ref{EqQmany}) we can find the inverse function $\gamma(q)$, which expresses the fraction of intercepted positions $\gamma$ dependent on the measured value of QBER $q$. 
Then we have
\begin{equation}\label{EqIEMany}
	I_{\rm E}(q)=\gamma(q)\left[1-\sum_{j=0}^{M-1}p_j(\gamma(q))h\left(\cos^2\left(\frac{\pi j}{2M}\right)\right)\right].
\end{equation}
To calculate the secret fraction, (\ref{EqIEMany}) should be substituted into (\ref{EqRate}).

It is useful to calculate the Eve's information in case $N\to\infty$ (also $l'\to\infty$ since $N=2^{l'}$). In this case, $n_{\rm correct}(\gamma)/n=1/2$ (see the Remark~\ref{RemCorrOpt} in Appendix~C and Eq.~(\ref{EqLegPatt}) in Apendix~A). 
Then we arrive at the following expression:
\begin{equation}
	q=\gamma\left[\frac14-\frac1M\sum_{j=0}^{M-1}\cos\left(\frac{2\pi j}{M}\right)\right]=\frac\gamma4.
\end{equation}
The mean Eve's information is then
\begin{equation}
\begin{split}
	I_{\rm E}(q){=}4q\left[1-\frac{1}{M}\sum_{j=0}^{M-1}h\left(\cos^2\left(\frac{\pi j}{2M}\right)\right)\right]\equiv4q\zeta(M),
\end{split}
\end{equation}
where $\zeta(M)$ is a decreasing function of $M$, 
\begin{equation}
\lim\limits_{M\to\infty}\zeta(M)\approx0.4427.
\end{equation}
The function $\zeta(M)$ is shown in Fig.~\ref{fig:zeta}.
\begin{figure}[t]
\begin{centering}
\includegraphics[scale=.45]{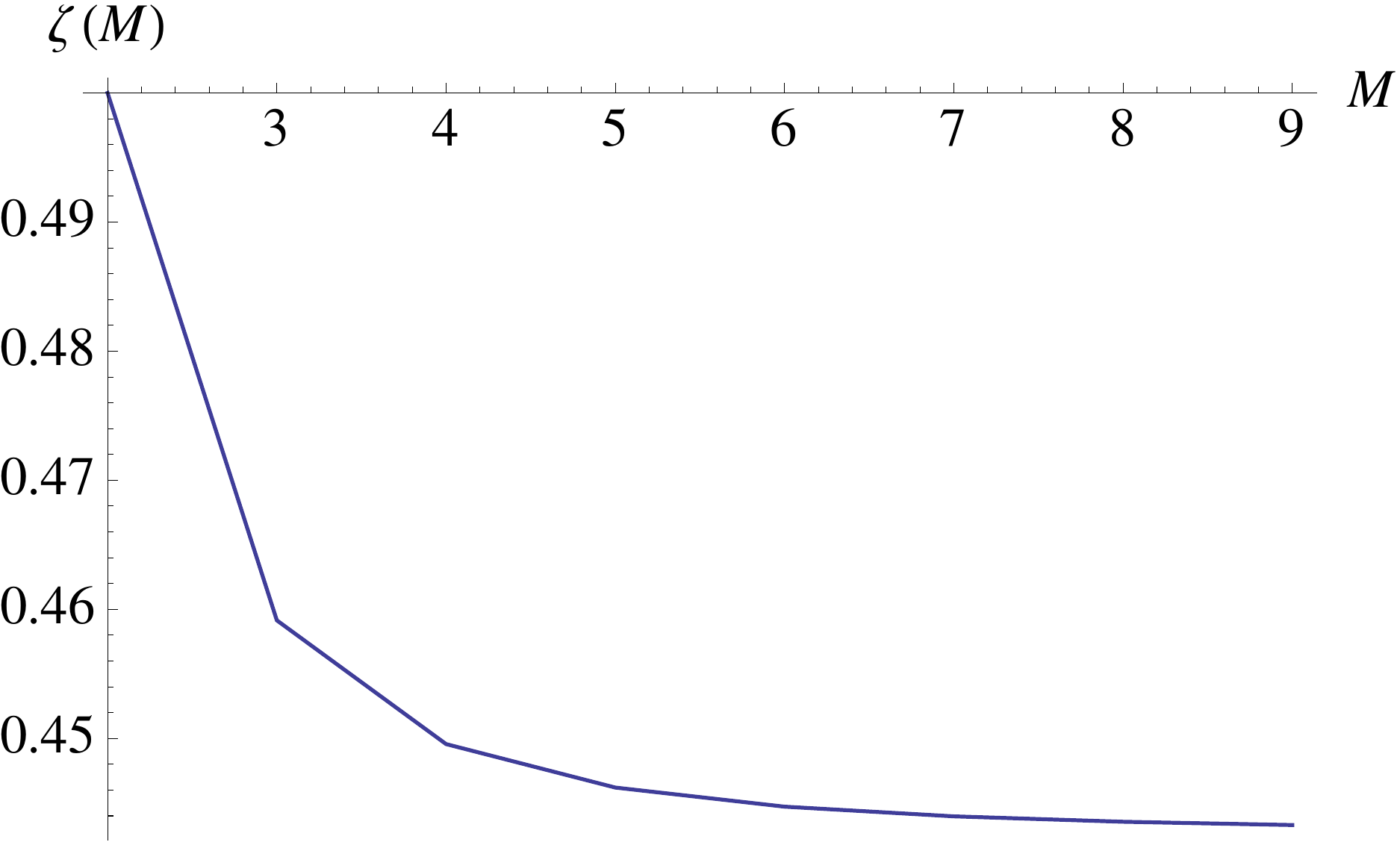}
\end{centering}
\vskip -4mm
\caption
{
Function $\zeta(M)$.
}
\label{fig:zeta}
\end{figure}
Since $\zeta(M)<0.5$ for $M>2$, the multibasis version protocol has advantage over the BB84 with pseudorandom sequence of bases.

Finally, we arrive at the following expression for the secret fraction:
\begin{equation}
	R=1-f(q)h(q)-I_{\rm E}(q)=1-f(q)h(q)-4q\zeta(M).
\end{equation}

\subsection{Numerical comparison}\label{SecNum}

We compare the secret key rates per bit of the raw key (secret fractions) for the multibasis PRB protocol (Eq.~(\ref{EqRate}) and Eq.~(\ref{EqIEMany})) with those for the BB84 protocol and the asymmetric BB84. The secret key rate per bit of the raw key for the BB84 protocol is given by
\begin{equation}\label{EqRateBB84}
	R(q)=\frac12[1-f(q)h(q)-2q],
\end{equation}
where the factor $1/2$ is due to sifting of a half of positions. 

We consider the following variation of the asymmetric BB84 protocol. 
For each pulse, Alice and Bob choose independently the basis $\{\ket0,\ket{\pi/2}\}$ (``the signal basis'') 
with the probability $1-p$ or the basis $\{\ket{\pi/4},\ket{3\pi/4}\}$ (``the test basis'') with the probability $p$. 
The first basis is used to establish the secret key, while the second one is used to detect the eavesdropping. 
The sifting rate is, thus, on average, $1-(1-p)^2$.  
Alice and Bob announce the bit values for positions encoded using the test basis and calculate the QBER in the test basis $q_\times$. 
Also, after the error correction step, they calculate the QBER in the first basis $q_+$ (see Sec.~\ref{sec:protocol}, step~(iv) of the protocol). 
We consider the case of the absence of actual eavesdropping, where the QBER is caused only by natural noise, then, on average, $q_+=q_\times=q$. 

The mean Eve's information per bit of the sifted key is $2q$. But now, for small $p$, we cannot neglect the statistical fluctuations. 
We can use, for example, the Hoeffding's inequality: if $X$ is a binomially distributed random variable with the probability of success $P$ and the number of trials $K$, then
\begin{equation}
	\Pr[X\geq(P+\delta)K]\leq e^{-2\delta^2K}=\varepsilon
\end{equation}
with $\varepsilon$ being the failure probability (probability that Eve is more lucky that we expect; its meaning is the same as $\varepsilon$ in the PRB protocol), or,
\begin{equation}\label{EqHoeffding}
	\delta(P,\varepsilon)=\sqrt{\frac{1}{2K}\ln\frac1\varepsilon}.
\end{equation}
Here we have $P=q$, $K=p^2N_{\rm r}$ (the mean number of positions received by Bob for which both Alice and Bob chose the test basis), where $N_{\rm r}\leq N$ denotes the number of pulses received by Bob (recall that $N$ is the number of pulses sent by Alice). Then
\begin{equation}\label{EqaBB84}
	R(q)=(1-p)^2\left[1-f(q)h(q)-
	2\left(q+\sqrt{\frac{1}{2p^2N_{\rm r}}\ln\frac1\varepsilon}\right)\right].
\end{equation}
So, a smaller $p$ leads to lower sifting, but also to higher statistical fluctuations of the potential Eve's information that we have to take into account. 
In the calculations we optimized (\ref{EqaBB84}) over $p$ for each value of $q$.

In the calculations, we use the following parameters: $L=N=10^{10}-33$, $l'=\log_2L\approx 16$, $m=10$ ($M=1024$ bases). 
To obtain function $\gamma(q)$ in (\ref{EqIEMany}), 
we have used formula (\ref{Eqncorroptlp}) with the failure probability $\varepsilon\approx 10^{-6}$ (more precisely, $\varepsilon/m=S/L$ for $S=1000$). 
The parameter $s$ in (\ref{Eqncorroptlp}) was set to $s=12$. 
The number of pulses sent by Alice for all protocols is equal to $L$. 
The number of pulses received by Bob is $N_{\rm r}=N$ if the quantum channel is lossless. 
For a lossy channel we have taken a realistic loss rate $N_{\rm r}/N=0.001$. 
The failure probability for aBB84 in (\ref{EqaBB84}) was also taken as $\varepsilon\approx 10^{-6}$. The results are given on Fig.~\ref{fig:res}.
 
\begin{figure}[t]
\begin{centering}
\includegraphics[scale=.45]{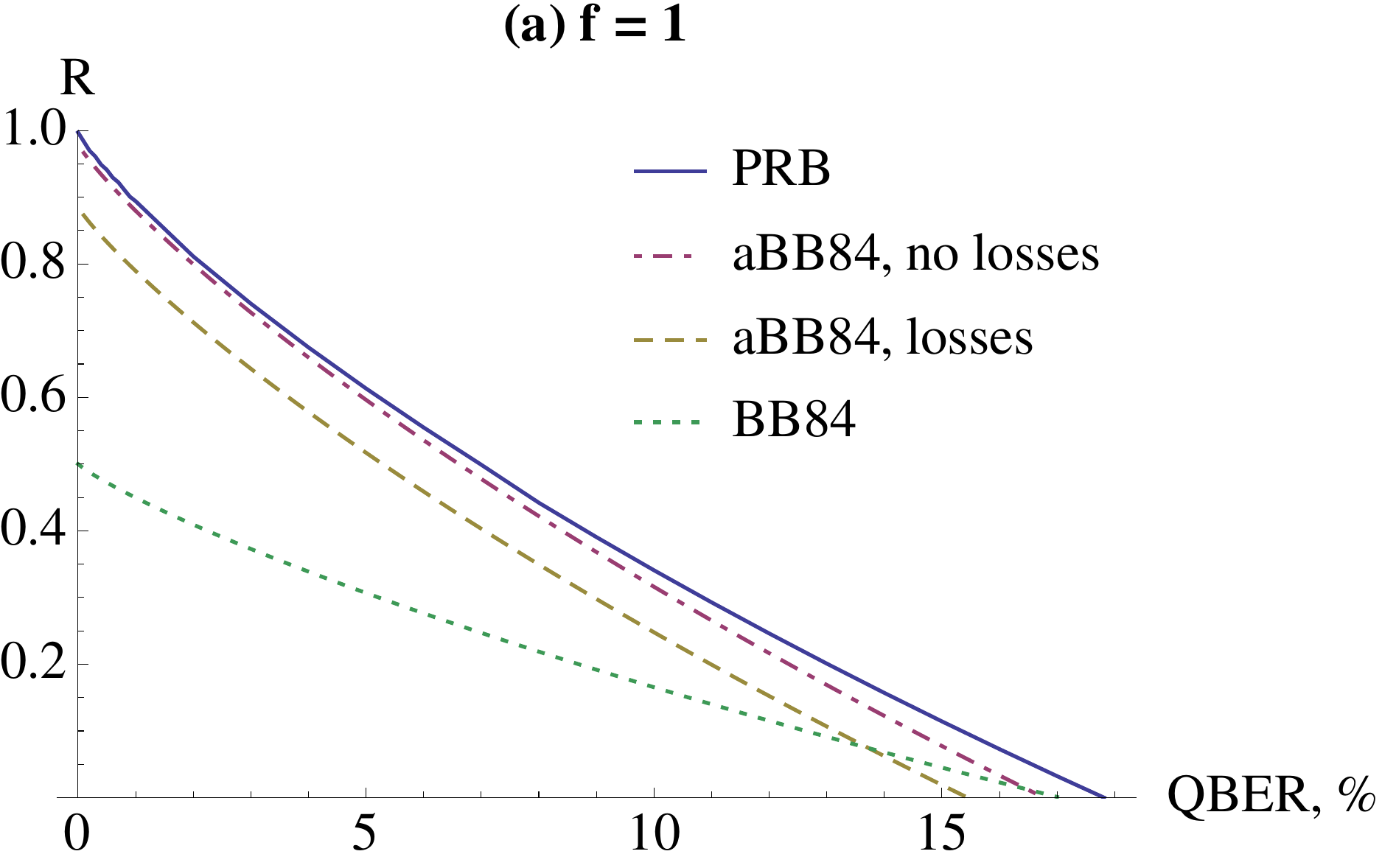}
\smallskip

\includegraphics[scale=.45]{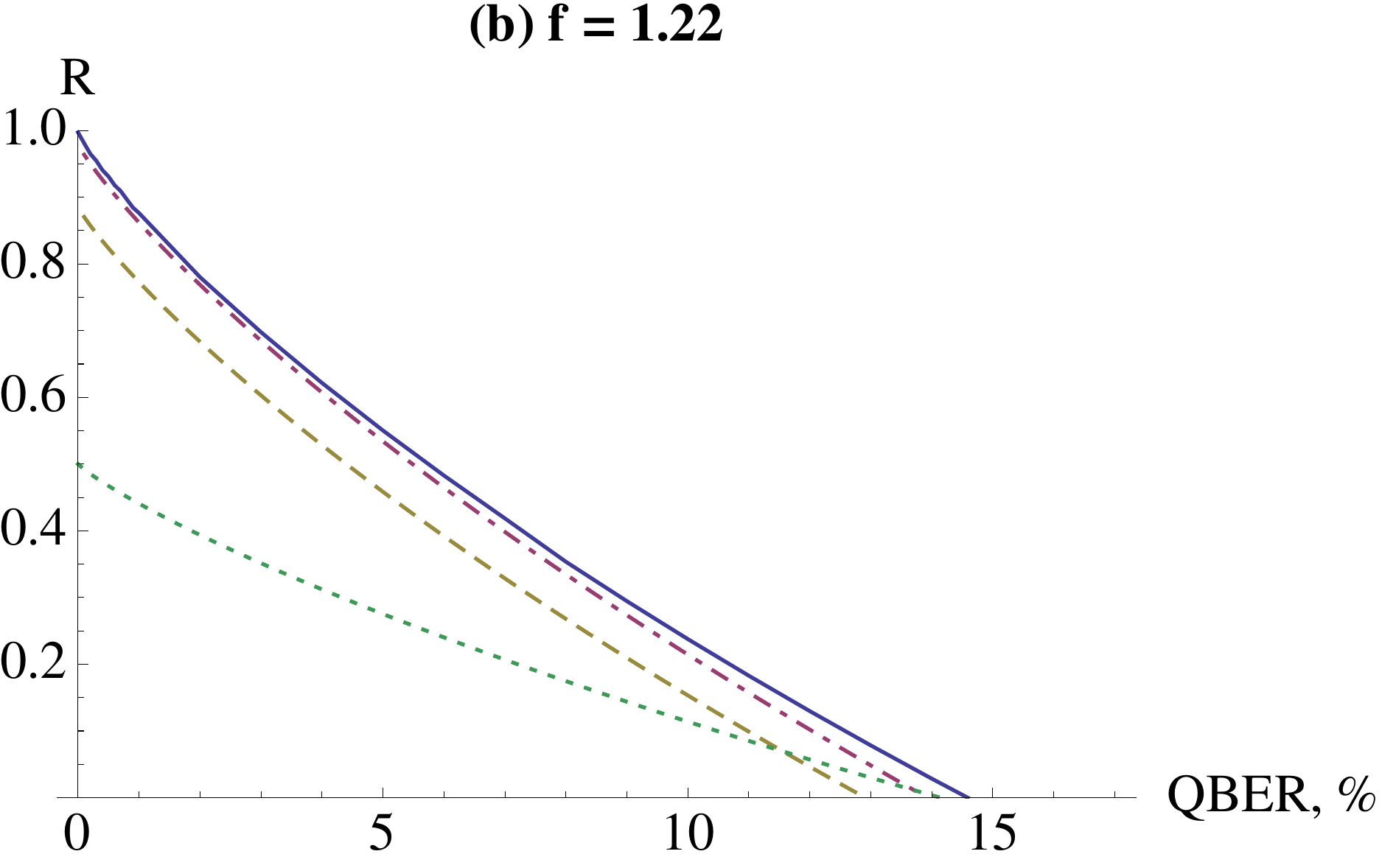}
\end{centering}
\vskip -4mm
\caption
{
Secret key rates (per bit of the raw key, before sifting) of the presented pseudorandom bases (PRB) protocol, the BB84 protocol and the asymmetric BB84 (aBB84) protocol with and without losses in the quantum channel, when the efficiency of error-correction achieves the theoretical limit $f=1$ (a) and for practically achievable efficiency $f=1.22$ (b). It can be seen that the considered protocol gives better secret key rates than the BB84 protocol and approximately the same rates as the asymmetric BB84 protocol.
}
\label{fig:res}
\end{figure}

It is clearly seen that the PRB protocol gives twice as large secret fraction as the BB84 protocol (due to absence of sifting). 
In PRB, the larger number of bases that Eve can correctly guess for the pseudorandom sequence is compensated by additional uncertainty for Eve caused by the use of many (instead of two) bases.
This results in approximately the same secret fractions for PRB and asymmetric BB84 for the lossless channels. 
However, if the channel is lossy, Alice and Bob have to increase $p$ and, hence, sifting rate, to collect large enough statistics for the test basis. 
In this case we can see that PRB gives slightly better results. 

Note that the losses do not decrease the secret fraction for the PRB protocol since the estimate of the number of bases correctly guessed by Eve is dependent on $L$ and not on $N_{\rm r}$.
Moreover, if optimal positions to attack are lost (see the end of Sec.~\ref{SecBB84PRB} 
for general comments on the optimal Eve's attack on the PRB protocol and Appendix~C for rigorous analysis), the losses even weaken Eve's attack .

\section{Photon-number splitting plus quantum states discrimination attack}\label{sec:pns}

We performed the analysis for an ideal case, where the light source is assumed to be single-photon. 
Practically, usually weak coherent pulses are used~\cite{Gisin,Scarani}. 
This gives possibilities to Eve to perform additional attacks, for example, photon number splitting (PNS) attack. 
In this attack it is assumed that  quantum technologies are fully accessible for Eve. 
Let us describe this attack for the BB84 protocol. 
Eve measures the number of photons in each pulse and, if the number of photons is at least two, takes one photon and save it in her quantum memory. 
After the announcement of the bases, she measures this photon in the known basis and, so, obtains a bit of information about the raw keys without disturbance. 
This potential Eve's information must be taken into account by Alice and Bob. 
Eve also can stop the single-photon pulses to increase the fraction of the multiphoton pulses (i.e., the number of pulses about which she can obtain full information without disturbance). 
To detect such actions, the so called decoy state method has been proposed and developed \cite{LoMa2005,Wang2005,MaLo2005,Ma2017,Trushechkin2016}. 
Its purpose allows one to obtain tight estimates on the number of single-photon pulses and the number of errors in these pulses. 

For the case of the proposed pseudorandom basis protocol, Eve can also perform the PNS attack. 
While in BB84 she  waits for the announcement of bases, here she waits for the moment when she gets full knowledge of the initial key and, hence, bases. 
To account for these attack, Alice and Bob can also use the decoy state method.

But now Eve can perform another type of attack, which we will refer to as ``photon number splitting plus quantum state discrimination'' (PNS+QSD) attack. 
Namely, Eve has a possibility to use multiphoton pulses to get knowledge of the initial key during the transmission of quantum states. 
Recall that all our analysis above was based on the assumption that Eve has zero information on the initial key during the transmission of quantum states and has to guess the bases. 
But now she can perform the following variant of the PNS attack. Again, she measures the number of photons in each pulse. 
If the number of photons in a pulse is at least three, she sends one photon to Bob (i.e., does not introduce disturbance) and takes two photons to her quantum memory. 
We have shown that a single photon without the knowledge of the raw key bit leaks no information about the basis (see Eqs.~(\ref{EqNoAprioriInf}) and~(\ref{EqChaoticState})). 
But this is not true if Eve has two photons in the same state. Let us analyse this attack.

Suppose that Eve has intercepted $n$ such double photons in positions $i_1,\ldots,i_n$. 
Then, to get knowledge of the initial key $x$, she has to distinguish between $2^{l+n}$ states
\begin{equation}
	\ket{\psi(k,\mathbf x)}=\bigotimes_{j=1}^n\ket{\varphi_{i_j}(k)+\frac\pi2x_{i_j}}^{\otimes 2},
\end{equation}
where $k\in\{0,1\}^l$, $\mathbf x=\{x_{i_1},\ldots,x_{i_n}\}\in\{0,1\}^n$.

Discrimination of quantum states (or hypothesis testing) is a famous problem in quantum information science~\cite{Holevo}. 
We will use the following lower bound on the success probability $p_{\rm succ}$ of guessing the correct quantum state (in our case -- correct $k$ and $\mathbf x$)~\cite{Montanaro}:
\begin{eqnarray*}
&&p_{\rm succ}\geq \frac1{2^{l+n}}\sum_k\sum_{\mathbf x}
\frac1{\sum_{k'}\sum_{\mathbf y}\braket{\psi(k,\mathbf x)|\psi(k',\mathbf y)}^2}\\
&&=
\frac1{2^{l+n}}\sum_k\sum_{\mathbf x}
\frac1{1+\sum_{k'\neq k}\sum_{\mathbf y}\braket{\psi(k,\mathbf x)|\psi(k',\mathbf y)}^2},
\end{eqnarray*}
\begin{eqnarray*}
&&\sum_{\mathbf y}\braket{\psi(k,\mathbf x)|\psi(k',\mathbf y)}^2\\&&=
\prod_{j=1}^n\left\lbrace
\cos^4[\varphi_{i_j}(k)-\varphi_{i_j}(k')]
+
\sin^4[\varphi_{i_j}(k)-\varphi_{i_j}(k')]
\right\rbrace
\end{eqnarray*}
We restrict the analysis to the case of two bases ($M=2$). 
The analysis of the multibasis case leads to more cumbersome calculations but qualitatively the same conclusions. 
If $k\neq k'$, then approximately a half of basis choices for the keys $k$ and $k'$ coincide. Hence, 
\begin{equation}
\sum_{\mathbf y}\braket{\psi(k,\mathbf x)|\psi(k',\mathbf y)}^2\approx 2^{-n/2}
\end{equation}
and
\begin{equation}
p_{\rm succ}\geq \frac1{1+(2^l-1)2^{-n/2}},
\end{equation}
i.e., Eve needs approximately $2l\ll N$ three-photon pulses to guess the secret key $k$ with a non-negligible probability. 
Then she can measure the pulses in correct bases without disturbance. Hence, the protocol crucially requires single-photon sources.

\section{Discussions and conclusions}\label{sec:conclusion}

In the present work, we have analyzed a new prepare-and-measure QKD protocol. 
It uses the pseudorandom sequence of bases generated by the legitimate parties of communications from a common initial secret key (seed). 
The use of a common pseudorandom sequence of bases allows one to avoid the sifting procedure and, hence, losing the half of the key. 
Moreover, since the bases of Alice and Bob are always the same, they can use more than two bases. 

The main result of this work is the calculation of the secret key rates of the new protocol for the intercept-resend attack presented on Fig.~\ref{fig:res}. 
The main technical ingredient is Appendix~C, 
where the mathematical tools of analysis of pseudorandom sequences in the context of quantum cryptography (where the adversary has unlimited computing power) are proposed. 
The main practical formulas derived in this appendix are (\ref{Eqncorropt}) and~(\ref{Eqncorroptlp}). 
They give upper bound on the elements of a pseudorandom sequence that can be correctly guessed by the eavesdropper with unlimited computing power.

We have obtained that, for single-photon sources, the new protocol gives twice as large secret key rates as the original BB84 protocol and in some cases gives slightly higher rates than the aBB84 protocol. 
However, we did some assumptions in favor of Eve (see, for example, Remark~\ref{RemNcorr}). 
More tight analysis, which requires the development of techniques of Appendix~C, probably, will lead to even more significant advantage of the proposed pseudorandom multibasis protocol. 
The protocol strongly requires a single-photon source of light.

The mathematical tools developed in Appendix~C can be used in different problems of quantum cryptography, 
for example, for rigorous estimation of how the use of pseudorandom sequences (instead of truly random ones) influence the security of the conventional prepare-and-measure QKD protocols, 
such as BB84, asymmetric BB84, etc. 
First steps in such analysis have been done in Ref.~\cite{Bouda2012}. 
The novelty of our approach is the suggestion that the pseudorandomness can be even turned to an advantage.

In general, future investigations of the power of classical pseudorandomness combined with quantum uncertainty in quantum cryptography are required. 
The fundamental difference with the consideration of pseudorandom sequences in conventional cryptography is that, 
in the latter case, one typically assumes the boundedness of the eavesdropper's computing power. 
For example, it is assumed that the eavesdropper cannot use the brute force attack to try all possible seeds for a PRNG. 
In contrast, in quantum cryptography we assume unlimited computing power of the eavesdropper. 
But, from the other side, the possibilities of the eavesdropper are limited by the quantum uncertainty principle. 
Thus, the analysis of quantum pseudorandom sequences may require novel mathematical methods.

\section{Acknowledgments}

We are grateful to A.V.~Duplinskiy, V.L.~Kurochkin, and V.E.~Ustimchik for fruitful discussions, and especially grateful to M.A.~Korolev and N.~L\"utkenhaus for valuable suggestions. 
The work of A.S.T. and E.O.K. was supported by the grant of the President of the Russian Federation (project MK-2815.2017.1).
The work of A.K.F. is supported by the RFBR grant (17-08-00742).

\setcounter{equation}{0}
\setcounter{section}{0}
\renewcommand{\theequation}{A\arabic{equation}}

\section*{Appendix A. Pseudorandomness property of the Legendre sequences}\label{sec:eps}

Here we formulate the keystone pseudorandomness property of the Legendre sequences which is exploited in the present analysis -- mainly, in Appendix~C. Note that 
\begin{equation}
	\overline a_{i+L}=\overline a_i, 
\end{equation}
i.e., $\{\overline a_i\}$ is a periodic sequence. 
We then denote $\mathbb Z_L=\{0,\ldots,L-1\}$ the residue ring with respect to integer addition and multiplication modulo $L$. 
For distinct elements $i_1,\ldots,i_s\in\mathbb Z_L$ and binary $b_1,\ldots,b_s$, denote $d_{i_1,i_2,\ldots,i_s}(b_1,b_2,\ldots,b_s)$ the number of $j\in\{0,\ldots,L-1\}$ such that
\begin{equation}
	\overline a_{j+i_1}=b_1,\quad \overline a_{j+i_2}=b_2,\quad \ldots\quad \overline a_{j+i_s}=b_s,
\end{equation}
i.e., the number of occurrences of the pattern 
\begin{equation}
	*\ldots*b_1*\ldots*b_2*\ldots*\ldots*b_s
\end{equation}
in one period. 
Here $*$ are ``do-not-care'' bits. 
In other words, we look for patterns with the bit values $b_1,\ldots,b_s$ on positions $i_1,\ldots,i_s$.
Here we do not care bit values on other positions.

We will use the following bounds \cite{Aladov1896,Moroz,Ding1998}: for all distinct $i_1,\ldots,i_s\in\mathbb Z_P$ and all binary $b_1,\ldots,b_s$:
\begin{subequations}\label{EqLegPatt}
\begin{equation}
	d_{i,j}(b_1,b_2)=
	\begin{cases}
		(L-3)/4,&(b_1,b_2)=(1,1)\\
		(L+1)/4,&(b_1,b_2)\neq(1,1),
	\end{cases}
\end{equation}
\begin{equation}
	-W(s)\leq d_{i_1,i_2,\ldots,i_s}(b_1,b_2,\ldots,b_s)-\frac L{2^s}\leq W(s)
\end{equation}
for $s\geq3$, where
\begin{equation}
	W(s)=\frac{\sqrt L[2^{s-1}(s-3)+2]+2^{s-1}(s+1)-1}{2^s}.
\end{equation}
\end{subequations}
For large $s$, $W$ can be approximated as $\sqrt L[(s-3)/2]+(s+1)/2$.  

\setcounter{equation}{0}
\setcounter{section}{0}
\renewcommand{\theequation}{B\arabic{equation}}

\section*{Appendix B. Guessing the seed for PRNG}\label{SecGuess}

In our analysis we assumed that, after all stages of the protocol and after the transmission of a message encrypted with the use of the distributed key, 
Eve can correctly guesses the initial secret key (the seed for the PRNG). 
In this Appendix, we derive bounds on the number of  qubits that Eve needs to intercept  for correct guessing of the seed. We show that this assumption is not too pessimistic.

Firstly, we specify assumptions on Eve's knowledge. 
Of course, Eve knows her measurement results of intercepted qubits $z_{i_1},\ldots,z_{i_n}$. 
We further assume that, during the stage of error correction, 
Eve discovers (along with a syndrome or other messages sent via error correction) the positions where she introduces errors and the values of the bits $x_i$ and $y_i$ in such positions. This is indeed true if the Cascade protocol for error correction is used, but may be too pessimistic in the case of the use of one-way error-correcting codes (for example, LDPC codes).
Let us also denote $c_i=x_i\oplus y_i$. 

Moreover, we assume that Eve may know a part of the message encoded with the key distributed by the protocol (``known plaintext attack''). 
Let $r$ bits of the secret key $(u_1,\ldots,u_r)$ are distributed; $r<N$ due to the key contraction in the privacy amplification stage. 
The last $l$ bits from this key are kept for the next session as a new initial key. The first $r-l$ bits are used for encryption of a message, for example, using one-time pad encryption. 
Eve may know a part of this message (or even the whole message) and, hence, the corresponding bits of the key $(u_1,\ldots,u_q)$, $q\leq r-l$. 
She can use this knowledge as well as her results of quantum measurements to guess the unknown part of the distributed key and, in particular, the initial key for the next session. 
The knowledge of the initial key for the next session gives her a possibility to obtain the key of the next session without introducing errors,
since she can also construct the sequence $\{\varphi_i(k)\}$ in the next session.

For definiteness, let us assume that Bob is the side that correct errors, so that after error correction Alice and Bob have the common key $x_1,\ldots,x_N$. 
If the Toeplitz hashing is used for privacy amplification, then $u_i$ are linear combinations (with respect to XOR) of $x_j$:
\begin{equation}\label{eq:pa}
	u_i=\sum_{j=1}^N t_{ij}x_j,\quad i=1,\ldots,r,
\end{equation}
where $t_{ij}$ are the elements of a Toeplitz matrix. 

Thus, Eve knows:
\begin{enumerate}[(i)]
	\item her measurement results $z_{i_1},\ldots,z_{i_n}$;
	\item the syndrome of the used error-correcting code (or parities of certain subsets of positions if an interactive error-correcting procedure like ``Cascade'' is used);
	\item whether she has introduced errors in the intercepted positions: $c_{i_1},\ldots,c_{i_n}$. Also she knows $x_{i_j}$ and $y_{i_j}$ if $c_{i_j}=1$;
	\item $q\leq r-l$ outputs of linear combinations (\ref{eq:pa}).
\end{enumerate}

For convenience of notations, let Eve attacks the first $r$ transmitted state and $i_j=j$. 
Also, to make the derivations simpler, we do the following modification of the protocol: 
let the angles $\varphi_i(k)$ in (\ref{eq:sequence}) are chosen from the set $\{\frac{\pi j}{2M}\}_{j=0}^{2M-1}$ rather than from $\{\frac{\pi j}{2M}\}_{j=0}^{M-1}$, $M\geq2$. 
Thus, we add an additional pseudorandom binary register which is responsible for an additional rotation of the angle over $\pi/2$. 
This does not alter the security properties of the protocol (if the initial key is also added by one bit) 
since a basis rotated over $\pi/2$ is in fact the same basis as the initial one up to the interchange of the assigned bit values 0 and 1. 
However, as we noted before, the bit value is supposed to be known to Eve after the transmission of an encrypted message. 
Such a modification is useless for practice since we should spend the initial secret key on the additional register, 
but it is useful for the purposes of the present section: this modification makes the situation more symmetric and simplifies the analysis. 

\begin{widetext} 

Thus, in expansion (\ref{EqRndRegs}), we have an additional register:
\begin{equation}\label{EqRndRegs0}
	\varphi_i(k)=\frac\pi2\sum_{j=0}^ma_i(k^{(j)})2^{-j}.
\end{equation}

Now we are going to estimate the probability of guessing the seed 
\begin{equation}
	p_{\text{guess}}=\max_k p(k|e_1,...,e_n), 
\end{equation}
where $e_i=(x_i,z_i,c_i)$ is Eve's knowledge on the $i$-th transmitted quantum state. 

Then we have
\begin{equation}\label{EqPk}
	p(k|e_1,...,e_n)=\frac{p(e_1,...,e_n|k)p(k)}{p(e_1,...,e_n)}=2^{-l}\frac{p(e_1|\varphi_1(k))\cdots p(e_n|\varphi_r(k))}{p(e_1,...,e_n)}.
\end{equation}
Let us find an expression for $p(e_i|\varphi_i)$:
\begin{equation}\label{EqEAlpha}
\begin{split}
	p(e_i|\varphi_i)&=p(x_i|\varphi_i)p(z_i|\varphi_i,x_i)p(c_i|\varphi_i,x_i,z_i)=\frac12p(z_i|\varphi_i,x_i)p(c_i|\varphi_i,x_i,z_i)\\
	&=\frac12\cos^2\left[\beta_i-\varphi_i+\frac\pi2(z_i-x_i)\right]\cos^2\left[\beta_i-\varphi_i+\frac\pi2(z_i-x_i-c_i)\right]\\
	&=\frac18\left\lbrace\cos\frac{\pi c_i}2+
	\cos\left[2(\beta_i-\varphi_i)+\pi(x_i-z_i)-\frac{\pi c_i}2\right]\right\rbrace^2\\
	&=\frac18\left\lbrace 1-c_i+
	\cos\left[2(\beta_i-\varphi_i)+\pi(x_i-z_i)-\frac{\pi c_i}2\right]\right\rbrace^2.
\end{split}
\end{equation}
Here we have used that $\cos\frac{\pi x}2=1-x$ if $x\in\{0,1\}$. 
Therefore, we arrive at the following expressions: 
\begin{equation}
	p(e_i)=\frac1M\sum_{j=0}^{M-1}p\left(e_i|\varphi_i=\frac{\pi j}M\right)=
	\begin{cases}
		3/{16},&c_i=0,\\
		1/{16},&c_i=1,
	\end{cases}
	\quad 
	p(c_i)=\sum_{x_i,z_i\in\{0,1\}}p(e_i)=
	\begin{cases}
		3/4,&c_i=0,\\
		1/4,&c_i=1.
	\end{cases}
\end{equation}

Substitution of Eq.~(\ref{EqEAlpha}) into Eq.~(\ref{EqPk}) yields
\begin{equation}\label{EqPk2}
\begin{split}
	p(k|e_1,...,e_n)&=\frac{2^{-l}}{p(e_1,...,e_n)}\prod_{i=1}^n \frac18\left\lbrace 1-c_i+\cos\left[2(\beta_i-\varphi_i)+\pi(x_i-z_i)-\frac{\pi c_i}2\right]\right\rbrace^2
\end{split}
\end{equation}

If $n\leq l'$ (recall that $l'$ is the length of the initial key for each register, while $l=ml'$ is the whole length of the initial key), 
then $\varphi_i$ are approximately independent and uniformly distributed on their domain. 
Hence, Eve's variables $e_1,\ldots,e_n$ are also approximately independent: 
$p(e_1,...,e_n)\approx p(e_1)\cdots p(e_n)$. By properties of the Legendre sequences, for every combination of angles $\varphi_i$, there exists a key $k$ generating this combination.
Then the maximization of Eq.~(\ref{EqPk2}) is equivalent to maximization of every separate term in the product in its right-hand side, i.e., maximization of (\ref{EqEAlpha}).

Let us maximize (\ref{EqEAlpha}) over $\varphi_i$. We have
\begin{equation}
	\max_{\varphi_i}p(e_i|\varphi_i)=
	\begin{cases}
		p(e_i|\beta_i)=1/2,&c_i=0,\:z_i=x_i,\\
		p(e_i|\beta_i+\pi/2)=1/2,&c_i=0,\:z_i\neq x_i,\\
		p(e_i|\beta_i+\pi/4)=1/8,&c_i=1.
	\end{cases}
\end{equation}
Other words, if Eve has not introduced an error and her measurement result $z_i$ has coincided with $x_i$, 
then her optimal guess is that Alice has chosen the same basis as Eve: $\varphi_i=\beta_i$. 
If Eve has not introduce an error, but her measurement result $z_i$ has not coincided with $x_i$, 
then her optimal guess is $\varphi_i=\beta_i+\pi/2$, i.e., Alice has chosen the basis different from the Eve's basis by $\pi/2$: the basis rotated by $\pi/2$ coincides with the initial basis up to a bit flip. 
Finally, if Eve has introduced an error, then her optimal guess is $\varphi_i=\beta_i+\pi/4$, which corresponds to a situation that yields an error with the maximal probability (1/2).

Putting it all together, we arrive at the following expression:
\begin{equation}\label{EqPguess}
	\max_k p(k|e_1,\ldots,e_n)=2^{-l}\left(\frac83\right)^{n_0}2^{n_1}=2^{-l+n_0\log\frac83+n_1}=2^{-l+3n_0+n_1-n_0\log3},
\end{equation}
where $n_0$ and $n_1$ are number of positions where $c_i=0$ and $c_i=1$ respectively. 
Indeed, $n_0+n_1=n$. Since $p(c_i=0)=3/4$ and $p(c_i=1)=1/4$, if $n$ is large, then $n_0\approx3n/4$, $n_1=n/4$. Then

\begin{equation}\label{EqPguessTyp}
\max_k p(k|e_1,\ldots,e_n)\approx
	2^{-l+n[\frac 34\log\frac83+\frac14]}=2^{-l+n[\frac 52-\frac{3\log3}4]}
	\approx 2^{-l+1.3n}.
\end{equation}
\end{widetext}

Recall that this derivation is valid if $n\leq l'$. But we are interested in the inverse case. 
Then the analysis is more complicated. 
Firstly, not every sequence of angles  $\{\varphi_i\}$ is possible: different angles are not independent, hence, 
maximization of the numerator in (\ref{EqPk2}) is not reduced to maximization its separate factors. 
Secondly, $e_1,\ldots,e_n$ are also not independent (as the measurement results of dependent quantum states).

Nevertheless, we will still use formula (\ref{EqPguess}) as an upper bound for the guessing probability. 
The arguments are as follows. 
The most advantageous situation for Eve is when the current angle $\varphi_i$ does not depend on the previous angles. 
In this case, the measurement gives Eve more information than the measurement in the case when Eve already has partial information on $\varphi_i$. 
Thus, in favour of Eve, we treat $\varphi_i$  independent from each other even if $r>l'$ and use formulas (\ref{EqPguess}) and (\ref{EqPguessTyp}). 
Numerical experiments confirm that the validity of these formulas as upper bounds. 
Then Eve can guess a key in approximately (lower bound):
\begin{equation}\label{EqTimeGuess}
	n=\frac l{\frac 52-\frac{3\log3}4}\approx 0.76\,l.
\end{equation}
Thus, Eve needs of order $l$ intercepted positions to correctly guess the initial secret key. Our numerical experiments with short enough keys (up to $l'=10$, which allows one to explicitly implement the proposed maximum likelihood method) and $m\leq8$ show that formula (\ref{EqTimeGuess}) is adequate as a rough estimate at least for the subkeys $k^{(0)}$ and $k^{(1)}$ that govern the highest-order (i.e., most important) bits in the binary expansion of angles (\ref{EqRndRegs0}). One needs much more iterations to correctly guess the lowest-order bits since  close quantum states are hard to distinguish. From the other side, lowest-order registers are less important.  Hence, the assumption that Eve gets a knowledge of the initial key after the transmission of the ciphertext (provided that she knows the plaintext) seems to be not too pessimistic.

\bigskip

\setcounter{equation}{0}
\setcounter{section}{0}
\renewcommand{\theequation}{C\arabic{equation}}

\section*{Appendix C. Guessing in pseudo-random binary sequences}\label{SecGuessing}

In this Appendix we obtain an upper bound for the number of correctly guessed bits in a certain class of binary pseudo-random  sequences.
We assume that Eve has access to unlimited computing power, and we design an optimal attack for Eve.

\begin{widetext}
Let us introduce assumptions about PRNG. Let $\{a_i(k)\}_{i=1}^\infty$ be a periodic sequence with the period $L$ for any $k$, i.e., $a_{i+L}(k)=a_i(k)$. 
The set of keys is $\mathcal K=\mathbb Z_L=\{0,\ldots,L-1\}$, the residue ring with respect to integer addition and multiplication modulo $L$. 
For distinct keys $k_1,\ldots,k_s$ and binary $b_1,\ldots,b_s$, denote 
\begin{equation}
	A_{b_1 \ldots b_s}(k_1,\ldots,k_s)=\{i\in\mathbb Z_L|\,a_i(k_1)=b_1,\ldots,a_i(k_s)=b_s\}.
\end{equation}

Let us assume that there exists $S\geq2$ such that, for all distinct keys $k_1,\ldots,k_S$ and for all binary $b_1,\ldots,b_S$,
\begin{equation}\label{EqPattIdeal}
	|A_{b_1,\ldots,b_s}(k_1,\ldots,k_s)|=\frac L{2^S}.
\end{equation}
It is also assumed that $L$ is divided by $2^S$.

Eve chooses the fraction $0<\gamma\leq1$ of positions that she will try to guess, i.e., she will try to guess $\gamma L$ positions in a period. 
Her aim is to choose the positions to maximize the fraction of the guessed outcomes. 
If Eve attacks all $L$ positions, then, due to Eq.~(\ref{EqPattIdeal}), she guesses exactly a half of positions, which is expected when the sequence is truly random. 
We are going to prove the upper bound for the case  $0<\gamma<1$.

\begin{theorem}\label{Th}
Let the pattern distribution satisfy (\ref{EqPattIdeal}) and we try to guess $n=\gamma L$ positions in a period. 
Then the number of correctly guessed bits does not exceed
\begin{equation}\label{Eqncorroptideal}
	n_{\rm correct}(\gamma)=L\left\lbrace P_{s-1}(r)+\frac{s-r-1}s(\gamma-2P_{s}(r))\right\rbrace
\end{equation}
with the probability at least $1-s/|\mathcal K|$, for every $2\leq s\leq S$. 
Here $P_s(t)=\Pr[X_s\leq t],$ where $X_s$ is a binomially distributed random variable with the number of experiments $s$ and the success probability in one experiment 1/2 (i.e., $P_s(t)$ is a cumulative distribution function), and $r$ is the integer such that $P_{s}(r)\leq\gamma/2$ but $P_{s}(r+1)>\gamma/2$.
\end{theorem}

\begin{proof}
Let us consider $s$ arbitrary distinct keys $k_1,\ldots,k_s$. Let $T\subset\mathbb Z_L$ be the set of positions chosen by Eve, $|T|=\gamma L$. 
Denote $n_{b_1\ldots b_s}=|A_{b_1\ldots b_s}(k_1,\ldots,k_s)\cap T|$. 

Let it be known that the actual key $k$ is one of the keys $k_1,\ldots,k_s$. 
Let $n^{(i)}_{\rm correct}$ is the number of our correct guesses provided that the actual key is $k_i$, $i=0,\ldots,s-1$. 
We try to maximize $\min(n^{(1)}_{\rm correct},\ldots,n^{(s)}_{\rm correct})$, i.e., the guaranteed number of correct guesses.  

Let the set $T$ be fixed. 
For each position $j\in T$, if the majority of the values $a_j(k_1),\ldots,a_j(k_s)$ is 0 (1), then the optimal guess $a^{\rm guess}_j$ is also equal to 0 (1), i.e.,
\begin{equation}
	a^{\rm guess}_j=
	\begin{cases}
		0,&{\rm HW}(a_j(k_1),\ldots,a_j(k_s))\leq s/2,\\
		1,&\text{otherwise},
	\end{cases}
\end{equation}
where ${\rm HW}(b_1,\ldots,b_s)$ is the Hamming weight of the vector $(b_1,\ldots,b_s)$. 
Then
\begin{equation}\label{Eqmi}
	n^{(i)}_{\rm correct}=
	\sum_{
	\begin{smallmatrix}
	(b_1,\ldots,b_s):\\
	{\rm HW}(b_1,\ldots,b_s)\leq s/2,\\
	b_i=0
	\end{smallmatrix}
	}
	n_{b_1\ldots b_s}
	+
	\sum_{
	\begin{smallmatrix}
	(b_1,\ldots,b_s):\\
	{\rm HW}(b_1,\ldots,b_s)> s/2,\\
	b_i=1
	\end{smallmatrix}
	}
	n_{b_1\ldots b_s}
\end{equation}
for $i=1,\ldots,s$. 

Thus, we have the following optimization problem with respect to $2^s$ integers $n_{b_1\ldots b_s}$:
\begin{equation}\label{EqOptProbl}
	\left\lbrace
	\begin{aligned}
	&\min(n^{(1)}_{\rm correct},\ldots,n^{(s)}_{\rm correct})\to\max,\\
	&\sum_{(b_1,\ldots,b_s)}n_{b_1\ldots b_s}=\gamma L,\\
	&n_{b_1\ldots b_s}\leq 2^{-s}L,\quad \forall (b_1,\ldots,b_s),
	\end{aligned}\right.
\end{equation}
where the last condition is a consequence of (\ref{EqPattIdeal}).
From the symmetry of the problem we can put
\begin{equation}
n_{b_1\ldots b_s}=n_{{\rm HW}(b_1,\ldots,b_s)}
\end{equation}
(that is, only the number of keys $k_i$ such that $a_j(k_i)=a^{\rm guess}_j$ for a certain position $j$ matters). 
So, $n^{(1)}_{\rm correct}=\ldots=n^{(s)}_{\rm correct}=n_{\rm correct}$. Denote also $\nu_t=n_t/L$ and $\nu_{\rm correct}=n_{\rm correct}/L$. 
 
The number of vectors $(b_1,\ldots,b_s)$ with the Hamming weight $t$ is equal to $\begin{pmatrix}s\\t\end{pmatrix}$ (a binomial coefficient). If we have a constraint $b_i=1$ for fixed $i$ (like in the summation in (\ref{Eqmi})), then the number of vectors with the Hamming weight $t$ is equal to $\begin{pmatrix}s-1\\t-1\end{pmatrix}$. If we have a constraint $b_i=0$ for fixed $i$, then the number of vectors with the Hamming weight $t$ is equal to $\begin{pmatrix}s-1\\t\end{pmatrix}$.
Thus,
\begin{equation}
\nu_{\rm correct}=\sum_{t=0}^{\lfloor s/2\rfloor}
\begin{pmatrix}s-1\\t\end{pmatrix}
\nu_t
+
\sum_{t=\lfloor s/2\rfloor+1}^s
\begin{pmatrix}s-1\\t-1\end{pmatrix}
\nu_t,
\end{equation}
and the optimization problem (\ref{EqOptProbl}) is reduced to
\begin{equation}\label{EqOptProbl2}
\left\lbrace
\begin{aligned}
&\nu_{\rm correct}=\sum_{t=0}^{\lfloor s/2\rfloor}
\begin{pmatrix}s-1\\t\end{pmatrix}
\nu_t
+
\sum_{t=\lfloor s/2\rfloor+1}^s
\begin{pmatrix}s-1\\t-1\end{pmatrix}
\nu_t
\to\max,\\
&\sum_{t=0}^{s}\begin{pmatrix}s\\t\end{pmatrix}\nu_t=\gamma,\\
&\nu_t\leq 2^{-s},\quad t=0,\ldots,s.
\end{aligned}
\right.
\end{equation}
Obviously, an optimal choice is to assign the maximally possible value $2^{-s}$ to $\nu_t$ with $t$ close to $0$ or $s$. This means that we prefer positions where the guessed value is true for large number of keys. Other words, we primarily try to maximize $\nu_0$ and $\nu_s$, then try to maximize $\nu_1$ and $\nu_{s-1}$, and so on. The restriction of this process is the first constraint in (\ref{EqOptProbl2}). Denote $r$ the minimal integer such that
\begin{eqnarray}
	&&\sum_{t=0}^r\begin{pmatrix}s\\t\end{pmatrix}2^{-s}=\sum_{t=s-r}^s\begin{pmatrix}s\\t\end{pmatrix}2^{-s}\leq\frac{\gamma}2, \mbox { but } \label{Eqr1} \\ 
	&&\sum_{t=0}^{r+1}\begin{pmatrix}s\\t\end{pmatrix}2^{-s}=\sum_{t=s-r-1}^s\begin{pmatrix}s\\t\end{pmatrix}2^{-s}>\frac{\gamma}2 \label{Eqr2}.
\end{eqnarray}
Eqs. (\ref{Eqr1}) and (\ref{Eqr2}) can be rewritten as $P_{s}(r)\leq\gamma/2$ and $P_s(r+1)>\gamma/2$. 

Thus, $\nu_t$ are set to maximally possible value $2^{-s}$ for $t\leq r$ and $t\geq s-r$. Then, $\nu_{r+1}$ and $\nu_{s-r-1}$ are assigned by as larger value as possible: 
\begin{equation}
\nu_{r+1}=\nu_{s-r-1}=\left(\frac\gamma2-P_s(r)\right)
\begin{pmatrix}s\\r+1\end{pmatrix}^{-1}.
\end{equation}
Other $\nu_t$ (i.e., for $r+2<t<s-r-2$) are set to zero. The optimal value of the target function $\nu_{\rm correct}$ is
\begin{equation}
\begin{split}
\nu_{\rm correct}&=
\sum_{t=0}^r
\begin{pmatrix}s-1\\t\end{pmatrix}2^{-s}+
\begin{pmatrix}s-1\\r+1\end{pmatrix}
\begin{pmatrix}s\\r+1\end{pmatrix}^{-1}
\left(\frac\gamma2-P_s(r)\right)\\
&+
\sum_{t=s-r}^s
\begin{pmatrix}s-1\\t-1\end{pmatrix}2^{-s}+
\begin{pmatrix}s-1\\s-r-1\end{pmatrix}
\begin{pmatrix}s-r-1\\r+1\end{pmatrix}^{-1}
\left(\frac\gamma2-P_s(s-r)\right)\\
&=
P_{s-1}(r)+\frac{s-r-1}{s}(\gamma-2P_s(r)).
\end{split}
\end{equation}

This result means that, for any $s$ keys, we cannot choose $\gamma L$ positions such that more than $n_{\rm correct}$ guesses given by formula (\ref{Eqncorroptideal}) are correct for all keys. 
This means that the number of correct guesses cannot be larger than $n_{\rm correct}$ with the probability at least $1-s/|\mathcal K|$.
\end{proof}

\begin{remark}\label{RemCorrOpt}
If, in (\ref{Eqncorroptideal}), $s\to\infty$, then, by definition of $r$, we have $P_s(r)\to\gamma/2$ and $P_{s-1}(r)\to\gamma/2$, so, $n_{\rm correct}\to L\gamma/2$. Recall that $L\gamma$ is the number of terms we try to guess, i.e., the fraction of correct guesses tends to a 1/2, as in the case of truly random sequences.
\end{remark}

Now let us relax condition (\ref{EqPattIdeal}).

\begin{corollary}\label{CorGuess}
Let the pattern distribution satisfy 
\begin{equation}\label{EqPatt}
|A_{b_1,\ldots,b_s}(k_1,\ldots,k_s)|=\frac L{2^s}+W(s)
\end{equation}
for all $s$ from some range, where $W(s)$ is some function. We try to guess $n=\gamma L$ positions in a period. Then, with the probability at least $1-\varepsilon$, the number of correctly guess bits does not exceed
\begin{equation}\label{Eqncorropt}
n_{\rm correct}(\gamma,\varepsilon)=L'
\left\lbrace P_{s-1}(r)+\frac{s-r-1}s(\gamma'-2P_{s}(r))\right\rbrace,
\end{equation}
where $s=\varepsilon|\mathcal K|$. Here
\begin{equation}
L'=L'(s)=L\left(1+\frac{2^sW(s)}L\right),
\qquad
\gamma'=\gamma'(s)=\gamma\left(1+\frac{2^sW(s)}L\right)^{-1},
\end{equation}
and $r$ is the integer such that $P_{s}(r)\leq\gamma'/2$ but $P_{s}(r+1)>\gamma'/2$.
\end{corollary}

Note that, for the PRNG based on the Legendre symbol, (\ref{EqPatt}) is satisfied due to (\ref{EqLegPatt}) and (\ref{EqLegendrePRNG}).

\begin{proof}
The proof is the same, but optimization problem (\ref{EqOptProbl}) is modified into
\begin{equation}\label{EqOptProbl3}
\left\lbrace
\begin{aligned}
&\min(n^{(1)}_{\rm correct},\ldots,n^{(s)}_{\rm correct})\to\max,\\
&\sum_{(b_1,\ldots,b_s)}n_{b_1\ldots b_s}=\gamma L=\gamma'L',\\
&n_{b_1\ldots b_s}\leq 2^{-s}L+W(s)=2^{-s}L',\quad \forall (b_1,\ldots,b_s).
\end{aligned}
\right.
\end{equation}
So, formula (\ref{Eqncorropt}) holds with the substitutions of $L$ and $\gamma$ by $L'$ and $\gamma'$.
\end{proof}

If $s>\lceil L\rceil=l$, formulas (\ref{Eqncorroptideal}) and (\ref{Eqncorropt}) may be too optimistic for Eve. But we can obtain more tight bounds.

\begin{corollary}\label{CorGuess2}
Let (\ref{EqPatt}) be satisfied for some $s$. Other conditions are as in Corollary~\ref{CorGuess}. Then, with the probability at least $1-\varepsilon$, the number of correctly guess bits does not exceed 
\begin{equation}\label{Eqncorroptlp}
n_{\rm correct}(\gamma,\varepsilon)=\nu_{\rm correct}^*L'.
\end{equation}
Here $\nu_{\rm correct}^*$ is the solution of the following linear programming problem for $S=\varepsilon|\mathcal K|$:
\begin{equation}\label{EqOptProbl4}
\left\lbrace
\begin{aligned}
&\nu_{\rm correct}=
\sum_{t=0}^{\lfloor S/2\rfloor}
\begin{pmatrix}S-1\\t\end{pmatrix}\nu_t+
\sum_{t=\lfloor S/2\rfloor+1}^{S-1}
\begin{pmatrix}S-1\\t-1\end{pmatrix}\nu_t\to\max\\
&\sum_{t=0}^{S}\begin{pmatrix}S\\t\end{pmatrix}\nu_t=\gamma'\\
&\sum_{t=0}^{S-s}
\begin{pmatrix}S-s\\t\end{pmatrix}
\nu_{h+t}\leq 2^{-s},\quad h=0,\ldots,s,
\end{aligned}
\right.
\end{equation}
with the agreement $\nu_t=\nu_{S-t}$, i.e., the actual number of variables in the optimization problem is 
$\lfloor (S+1)/2\rfloor$.
\end{corollary}

\begin{proof}
Let us consider $S$ arbitrary keys $k_1,\ldots,k_S$, but condition (\ref{EqPatt}) is satisfied for $s\leq S$. 
A generalization of optimization problem (\ref{EqOptProbl3}) to this case is

\begin{equation}\label{EqOptProbl5}
\left\lbrace
\begin{aligned}
&\min(n^{(1)}_{\rm correct},\ldots,n^{(s)}_{\rm correct})\to\max,\\
&\sum_{(b_1,\ldots,b_S)}n_{b_1\ldots b_S}=\gamma'L',\\
&\sum_{
\begin{smallmatrix}
(b_1,\ldots,b_S):\\
b_{i_1}=c_1,\ldots,b_{i_s}=c_s
\end{smallmatrix}
}
n_{b_1\ldots b_S}
\leq 2^{-s}L',\quad 
\forall (c_1,\ldots,c_s), \:
 1\leq i_1<\ldots<i_s\leq S,
\end{aligned}
\right.
\end{equation}
where
\begin{equation}
n^{(i)}_{\rm correct}=\sum_{
\begin{smallmatrix}
(b_1,\ldots,b_S):\\
{\rm HW}(b_1,\ldots,b_S)\leq S/2,\\
b_i=0
\end{smallmatrix}
}
n_{b_1\ldots b_S}
+
\sum_{
\begin{smallmatrix}
(b_1,\ldots,b_S):\\
{\rm HW}(b_1,\ldots,b_S)> S/2,\\
b_i=1
\end{smallmatrix}
}
n_{b_1\ldots b_S}
\end{equation}

Again, by the symmetry, we can put 
\begin{equation}
n_{b_1\ldots b_S}=n_{{\rm HW}(b_1,\ldots,b_S)}
\end{equation}
and $n_t=n_{S-t}$, so, the problem is reduced to (\ref{EqOptProbl4}) (where $\nu_t=n_t/L'$ and $\nu_{\rm correct}=n_{\rm correct}/L'$). Let us comment the last set of constraints. If ${\rm HW}(c_1,\ldots,c_s)=h$, then the left-hand side of the last constraint in (\ref{EqOptProbl5}) is reduced to
\begin{equation}
\sum_{t=h}^{h+S-s}
\begin{pmatrix}S-s\\t-h\end{pmatrix}
\nu_t
=
\sum_{t=0}^{S-s}
\begin{pmatrix}S-s\\t\end{pmatrix}
\nu_{h+t},
\end{equation}
which coincides with that of (\ref{EqOptProbl4}).
\end{proof}

\end{widetext}

\begin{remark}
The parameter $s$ in (\ref{EqOptProbl4}) is, in fact, an optimization parameter. If we use the Legendre sequences as PRNG, then, according to (\ref{EqLegPatt}), $s$ should be taken between $\sqrt l$ and $l$: smaller $s$ leads to less tight bounds, while larger $s$ lead to large deviations $W(s)$ in (\ref{EqLegPatt}), which also lead to less tight bounds.
\end{remark}

The comparison of the results of formulas (\ref{Eqncorropt}) and (\ref{Eqncorroptlp}) for Legendre sequences (see Sec.~\ref{SecLegendre} and Appendix~A) is given on Fig.~\ref{fig:deltaprng}. 
\begin{figure}[t]
\begin{centering}
\includegraphics[scale=.45]{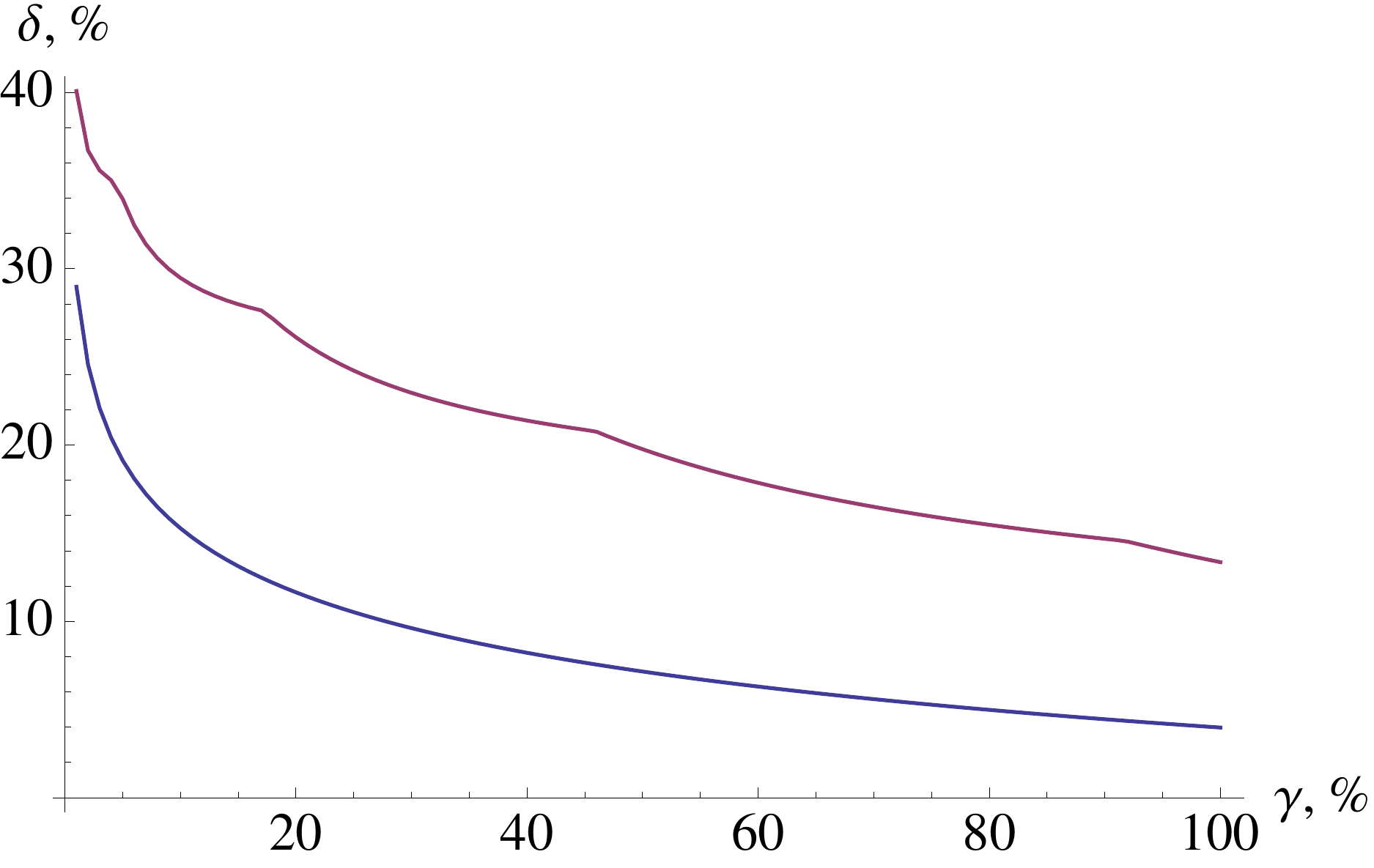}
\end{centering}
\vskip -4mm
\caption
{
Results of analysis of the Legendre sequences in terms of $\delta(\alpha)=\nu_{\rm correct}(\alpha)-1/2$ by formulas (\ref{Eqncorropt}) (top curve) and (\ref{Eqncorroptlp}) (bottom curve) for $L=10^{10}-33$, $\log_2L\approx 16$, $s=12$, and $S=1000$.
}
\label{fig:deltaprng}
\end{figure}
We took $L=10^{10}-33$, $\log_2L\approx 16$, $s=12$ (for both (\ref{Eqncorropt}) and (\ref{Eqncorroptlp})), and $S=1000$, $\varepsilon=S/L\approx10^{-7}$. 
We calculate the quantity $\delta(\gamma)=\nu_{\rm correct}(\gamma)-1/2$, i.e., deviation of $\nu_{\rm correct}$ from the mean value 1/2 in the case of random guessing. 

It is clearly seen that the results of (\ref{Eqncorroptlp}) are significantly better. 
We however note that the estimate for truly random sequences given by the Hoeffding's inequality (\ref{EqHoeffding}) for $P=1/2$, $K=n$, and $\varepsilon=10^{-7}$ yields 
$\delta\approx2.8\times10^{-4}=0.028\%$ even for $n=0.01N$ 
(i.e., $\gamma=0.01=1\%$). 
Thus, from the cryptanalyst's point of view (if it has enough computing power), 
an optimal guesses of elements of pseudorandom sequences gives much better results than the simple random guessing.

\end{document}